\documentclass[9pt,shortpaper,twoside,web]{ieeecolor}

\newif\ifproofs
\proofstrue
\newif \ifonlineversion
\onlineversiontrue
\usepackage{generic}

\usepackage{cite}

\usepackage{graphicx}
\usepackage{textcomp}

\usepackage{amsmath,amsfonts,amscd,amssymb,dsfont}
\usepackage{cleveref}
\usepackage{bm}
\usepackage{textcomp}
\usepackage{subfig}
\let\labelindent\relax
\usepackage{enumitem}

\usepackage{graphics} %

\usepackage{algorithm,algorithmicx}
\makeatletter
\renewcommand{\ALG@name}{Algo.}
\makeatother
\usepackage{algpseudocode}

\usepackage{pgfplots}
\pgfplotsset{compat=1.11}
\usepgfplotslibrary{fillbetween}
\usetikzlibrary{intersections}
\pgfdeclarelayer{bg}
\pgfsetlayers{bg,main}

\usepackage[colorinlistoftodos,bordercolor=orange,backgroundcolor=orange!20,linecolor=orange,textsize=scriptsize]{todonotes}
\usepackage{multirow}
\usepackage{pbox}
\usepackage{amsfonts}

\usepackage{stackengine}
\newcommand{\norm}[1]{\left\Vert #1\right\Vert}

\newcommand{\Proj}{\mathrm{P}}
\newcommand{\txt}{\textstyle}

\newcommand{\x}{x}
\newcommand{\xx}{\bm{x}}

\newcommand{\yy}{\bm{y}}

\newcommand{\zz}{\bm{z}}

\newcommand{\hf}{\hat{f}}
\newcommand{\g}{\mathbf{g}}
\newcommand{\bh}{\mathbf{h}}

\newcommand{\xag}{X}

\newcommand{\xxag}{\bm{\xag}}
\newcommand{\bxxag}{\bar{\bm{\xag}}}
\newcommand{\bsxxag}{\bar{\bm{\xag}}^*}
\newcommand{\yag}{Y}
\newcommand{\byag}{\bar{\yag}}

\newcommand{\byyag}{\bar{\bm{\yag}}}
\newcommand{\zag}{Z}

\newcommand{\bzzag}{\bar{\bm{\zag}}}

\newcommand{\rit}{\mathbb{R}}
\newcommand{\nit}{\mathbb{N}}

\newcommand{\X}{\mathcal{X}}
\newcommand{\FX}{\bm{\X}}
\newcommand{\FXA}{\FX_{\hspace{-2pt}\Aset}}
\newcommand{\Sxag}{\overline{\X}} %
\newcommand{\Y}{\mathcal{Y}}

\newcommand{\Aset}{\mathcal{A}}

\newcommand{\T}{\mathcal{T}}

\newcommand{\I}{\mathcal{I}}
\newcommand{\G}{\mathcal{G}}
\newcommand{\GA}{\mathcal{G}(\Aset)}

\newcommand{\M}{\mathcal{M}}
\newcommand{\N}{\mathcal{N}}

\newcommand{\ww}{\bm{w}}

\newcommand{\cc}{\bm{c}} %
\newcommand{\bv}{\bm{v}}%

\newcommand{\diamX}{R} %

\renewcommand{\th}{\theta} %
\renewcommand{\t}{t}

\newcommand{\n}{n}
\newcommand{\m}{m}

\newcommand{\nt}{_{\n,\t}}

\newcommand{\hxx}{\hat{\xx}}

\newcommand{\hxxag}{\hat{\xxag}}

\newcommand{\sxx}{\xx^*}  %
\newcommand{\sxxag}{\xxag^*}

\newcommand{\eqd}{ \mathrel{\ensurestackMath{\stackon[-1.5pt]{}{\overset{\smash{\scriptscriptstyle{\mathrm{def}}}}{=}}}}}
\newcommand{\dth}{\,\mathrm{d}\th}

\newcommand{\epopset}{^{\popset}}
\newcommand{\dset}{\delta} %
\newcommand{\mdset}{\overline{\delta}} %

\newcommand{\duti}{d} %
\newcommand{\mduti}{\overline{\duti}} %

\newcommand{\stgccvut}{{\alpha}} %

\newcommand{\Amat}{\bm{A}}
\newcommand{\bb}{\bm{b}}

\newcommand{\rlt}{\operatorname{ri}\,}
\newcommand{\rbd}{\operatorname{rbd}\,}

\newcommand{\aff}{\operatorname{aff}\,}

\newcommand{\bpsi}{\overline{\psi}}
\newcommand{\tG}{\tilde{\G}}

\newcommand{\pidx}{i} %
\newcommand{\popset}{\I} %
\newcommand{\popcard}{p} %
\newcommand{\hh}{\hspace{-2pt}}

\newcommand{\llam}{\bm{\lambda}}

\newcommand{\ux}{\underline{\x}}
\newcommand{\ox}{\overline{\x}}
\newcommand{\uux}{\bm{\ux}}
\newcommand{\oox}{\bm{\ox}}

\newcommand{\tr}{\top}

\newcommand{\paramvec}{\bm{v}}

\newcommand{\corProdGradHf}{\hat{H}} 
\newcommand{\corProdGradf}{H} 

\newtheorem{theorem}{Theorem}
\newtheorem{corollary}{Corollary}
\newtheorem{proposition}{Proposition}
\newtheorem{lemma}{Lemma}
\newtheorem{remark}{Remark}
\newtheorem{assumption}{Assumption}
\newtheorem{example}{Example}
\newtheorem{definition}{Definition}%

\Crefname{corollary}{Cor.}{Cors.}
\Crefname{equation}{Eq.}{Eqs.}
\Crefname{figure}{Fig.}{Figs.}
\Crefname{tabular}{Tab.}{Tabs.}
\Crefname{table}{Tab.}{Tabs.}
\Crefname{theorem}{Thm.}{Thms.}
\Crefname{definition}{Def.}{Defs.}
\Crefname{section}{Sec.}{Secs.}
\Crefname{proposition}{Prop.}{Props.}
\Crefname{assumption}{Asm.}{Asms.}
\Crefname{example}{Ex.}{Exs.}

\Crefname{appsec}{Appendix}{Appendices}

\begin{document}

\title{\LARGE \bf
Efficient Estimation of Equilibria in \\ Large Aggregative Games with  Coupling Constraints
}

\author{Paulin Jacquot, Cheng Wan, Olivier Beaude, Nadia Oudjane \thanks{Paulin Jacquot is with EDF R\&D (OSIRIS), Inria and \'Ecole polytechnique, CNRS, Palaiseau, France.} \thanks{\texttt{paulin.jacquot@polytechnique.edu}} \thanks{Cheng Wan, Olivier Beaude and Nadia Oudjane are with EDF R\&D (OSIRIS), Palaiseau, France (\texttt{\{cheng.wan, olivier.beaude,  nadia.oudjane\}@edf.fr}).}
\thanks{
This work was partially supported by the PGMO-ICODE  project ``Jeux de pilotage de flexibilit\'es de consommation \'electrique : dynamique et aspect composite". }
}

\maketitle

\begin{abstract} Aggregative games have many industrial applications, and computing an equilibrium in those games is challenging when the number of players is large. 
In the framework  of  atomic aggregative games with coupling constraints, %
 we show that variational Nash equilibria %
 of a large aggregative game can be approximated by a Wardrop equilibrium of an auxiliary population game of smaller dimension. Each population of this auxiliary game corresponds to a group of atomic players of the initial large game. 
 This approach enables an efficient computation of an approximated equilibrium, as the variational inequality characterizing the Wardrop equilibrium is of smaller dimension than the initial one. 
 This is illustrated on an example in the smart grid context.
\end{abstract}


\section{Introduction}
\paragraph{Motivation}
Aggregative games \cite{jensen2010aggregative} form a class of noncooperative games \cite{nisan2007algorithmic} in which each agent's objective is a function of her own action and the aggregate actions of all the agents.
These games find  practical applications---in particular for the subclass of splittable congestion games \cite{orda1993competitive}, in which resources are shared among agents, with  each resource having a cost function of the aggregate load onto it---in various fields such as traffic management \cite{ziegelmeyer2008road}, communications \cite{scutari2012monotone,altman2006survey} and electrical systems \cite{mohsenian2010autonomous,PaulinTSG17}.

The notion of Nash Equilibrium (NE) \cite{nash1950equilibrium} has emerged as the central solution concept in game theory.
However, the computation of an NE is considered a challenging problem: indeed, recent works have proved the theoretical complexity of the problem (PPAD-completeness  \cite{fabrikant2004complexity},%
\cite{klimm2018parametrized}). %
In continuous games, NEs can be characterized as solutions of Variational Inequalities (VI) \cite{facchinei2007finite}, a characterization we adopt throughout this work. 
 The efficiency of NEs computation depends on the dimension of these VIs, given by the number of players and of constraints.    
In some applications (transport, energy, etc), one may consider several thousands of heterogeneous agents. At this scale, computing an NE can be intractable. 
In the case where the game model involves coupling constraints between agents' actions, the concept generalizing the NE is simply referred to as Generalized Nash Equilibria (GNE) \cite{harker1991generalized}, and can be even harder to compute.
The consideration of coupling constraints has emerged from applications: for instance to model the capacity of roads and communication channels or, in energy, to model the maximal variations of  production plants  \cite{carrion2006computationally}. 
Estimating an NE as the probable outcome of the model is of main importance for the  operator of a system: for instance, in order to control some physical or managerial parameters and optimize the system (e.g. in  bi-level programming \cite{ColsonMS2007bilevel}). 

In this paper, we consider aggregative games with a finite but \emph{large} number of players. We propose a method to compute an approximation of Variational Nash Equilibria (VNEs) \cite{harker1991generalized}, a subset of GNEs in the presence of coupling constraints. 
The main idea of the proposed method is to reduce the dimension of the VIs characterizing VNEs:  players are divided into groups with similar characteristics. Then, each group is replaced by a homogeneous population of nonatomic (i.e. infinitesimal) players. 
Last, we compute a Variational Wardrop Equilibrium (VWE) \cite{wardrop1952some} in the obtained approximating nonatomic population game. %
The quality of the NE estimation depends on how well the characteristics (action sets and cost functions) of each homogeneous population approximate those of the atomic players it replaces.  

\paragraph{Related works}

The relation between NEs and WEs has been studied in different contexts:  \cite{haurie1985relationship} consider network congestion game with each origin-destination pair having $n$ players and show the convergence of NEs to a WE when $n$ goes to infinity. \cite{wan2012coalitions} generalizes this result in the framework of composite games where
nonatomic players and atomic splittable players coexist. 

A work particularly in line with the present paper is \cite{gentile2017nash} where the authors show, using  a variational inequalities approach, that the distance between a VNE and a VWE in aggregative games  converges to zero when the number of players tends to infinity. 
Their VWE corresponds to an equilibrium of the game where each atomic player is replaced by a population, as done in  \Cref{subsec:SVWEapproxVNE} of this paper. Their quantitative results show that the self impact of each player on the aggregate action is negligible when the number of players is large.
Our paper shows a different result: VNEs are approximated  by a VWE of a game with a \emph{reduced dimension}.
To this end, \Cref{subsec:SVWEapproxVNE} extends the results of \cite{gentile2017nash} to the subdifferentiable case. 
In \cite{PaulinWan2018routingGameCDC}, Jacquot and Wan show that, in congestion games with a continuum of heterogeneous players, the VWE can be approximated by a VNE of an approximating game with a finite number of players. 
In \cite{PaulinWan2018nonsmooth}, those results are extended to aggregative games  with nonsmooth cost functions. 
 The approach developed in the present paper is actually the inverse of the one taken in \cite{PaulinWan2018routingGameCDC} and \cite{PaulinWan2018nonsmooth}: here, the VWE in the auxiliary game serves as an approximation of a VNE of the original large game.

\paragraph{Main contributions}
\begin{itemize}[wide]
\item We define (\Cref{subsec:class}) an approximating population game, with smaller dimension but close enough to the original large game---quantified through the Hausdorff distance between action sets and between subgradients of players' objective functions;
\item we show (\Cref{th:pseudoVNE}) that VNEs are close to VWEs in large aggregative games, extending \cite[Thm.~1]{gentile2017nash} to subdifferentiable cost functions and general \emph{aggregatively monotone} games;
\item we show (\Cref{th:estimNE:main}) that a particular VWE of the  approximating reduced-dimension population game is close to any VNE of the original game with or without coupling constraints, and we provide an explicit expression of this approximation error;%
\item we provide a numerical illustration of our results (\Cref{sec:appliDRapprox}) based on a practical application: the decentralized charging of electric vehicles through demand response \cite{palensky2011demand}. %
This example illustrates the nondifferentiable case through piece-wise linear functions (``block rates tariffs'') and  coupling constraints (capacities and limited variations). 
This example shows that the proposed method is implementable and reduces the time needed to compute an equilibrium by computing its approximation (six times faster for an approximation with a normalized distance to the actual equilibrium of less than $2\%$).
\end{itemize}

\paragraph{Structure}
The remainder of this paper is organized as follows:  \Cref{sec:congestionModel} specifies the framework of aggregative games with coupling constraints, and recalls the notions of variational  equilibria and monotonicity for VIs, as well as several results on the existence and uniqueness of equilibria. 
\Cref{sec:approxRes} formulates the main results: \Cref{subsec:SVWEapproxVNE} shows that a VWE approximates VNEs in large games and then, \Cref{subsec:class} formulates the approximating population game with the approximation measures, and gives an upper bound on the distance between the VWE profile of the approximating game and an original VNE profile. 
\Cref{sec:appliDRapprox} presents a numerical illustration in the context of electricity demand response.

\section{Aggregative Games with Coupling Constraints}  \label{sec:congestionModel}

\subsection{Model and equilibria}
We consider the model of \emph{atomic aggregative} games:  the cost function of each player depends only on her own action and on the \emph{aggregate action}.
The term \emph{atomic} is opposed to the \emph{nonatomic} model where players are infinitesimal and have a negligible weight \cite{nisan2007algorithmic}. 
As the aggregate action of players is considered in this model, all players have their feasible action sets lying in the same space $\rit^T$ with $T\in \nit^*\eqd\nit\setminus\{0\}$ fixed. 
 The formal definition is as follows.
 \begin{definition}\label{def:atomicGame} 
An \emph{atomic aggregative game with coupling constraints} $\GA$ is defined by a finite set of players $\N \eqd \{1,\dots, \n , \dots , N\}$, a fixed $T \hh \in \nit^*$ and:
\begin{itemize}[leftmargin=*,wide,labelindent=-1pt]
\item for each player  $n  \in\N$, a set of feasible actions $\X_\n  \subset  \rit^T$; 
\item a set $\Aset\subset \rit^T$ defining the \emph{aggregative coupling constraint}: $\bxxag \in \Aset$ where $\bxxag \eqd \tfrac{1}{N}\sum_n \xx_n$ denotes the average action for individual actions of players $(\xx_n)_{n\in\N}$;
\item for each   $n   \in\N$, a cost function $f_\n\hh : \X_\n \hh\times\hh \Aset   \rightarrow \rit$, which defines the cost $f_n(\xx_n,\bxxag)$ for individual action $\xx_n$ and average action $\bxxag$.
\end{itemize}
The game is written as the tuple $\GA \eqd \big(\N, \FX, (f_n)_n, \Aset \big)$.
\end{definition}
We denote  by $\FX\eqd \X_1\times\dots \times \X_N$ the product set of action profiles,
and the set of feasible average profiles  by:\begin{equation*}
\resizebox{\columnwidth}{!}{$\Sxag\eqd \{ \xxag \in \rit^T  \hh : \forall \n \in \N ,  \exists \xx_\n  \in \X_\n \, \text{ s.t. }\txt\frac{1}{N}\sum_{\n \in\N } \xx_\n = \xxag  \}.$}
\end{equation*} 
Last, we denote the set of feasible action profiles satisfying the coupling constraint $\Aset$ by:
\begin{equation*} 
\FX(\Aset) \eqd \{\xx \in \FX: \txt\frac{1}{N}\sum_{\n \in\N } \xx_\n  \in  \Aset \} \ . 
\end{equation*}

A particular class of aggregative games with many applications is given by  \emph{splittable congestion games}\cite{orda1993competitive}: %
\begin{example} \label{def:splittableCongGame} A splittable congestion game is defined by a set of resources $\T\eqd \{1,\dots,T\}$ and a set of player $\N$ with actions set $(\X_n)_{n\in\N} \subset (\rit^\T)^\N$. Each player $n \in\N$ chooses an action $(\x\nt)_{t\in\T}\in \X_n$, where $\x\nt$ is the load of $ n$ on resource $t\in\T$.  Besides, each resource $t$ has a cost function $c_t: \rit \rightarrow \rit$,  while each player has an individual utility function $u_n:\X_n \rightarrow \rit$, which define the cost function $f_n  : \X_\n  \times \Sxag \rightarrow \rit$ of player $n\in\N$  as:
\begin{equation}\label{eq:cost_player_def}
\forall (\xx_n,\byyag) \hh \in \hh \X_n \hh \times \Sxag, f_\n (\xx_\n ,\byyag)  \eqd  \sum_{\t\in\T}\x\nt c_t(\byag_t) -u_\n (\xx_\n ) .
\end{equation} 
\end{example}

In aggregative games, the action $\xx_n$ of player $n$ appears both in the first and second argument (through $\bxxag$) of her cost function $f_n$. This leads us to consider for each $n\in\N$, the \emph{modified} cost function:
\begin{equation*}
 \hat{f}_\n (\xx_\n ,\byyag_{-\n}) \eqd f_\n (\xx_\n , \byyag_{-\n} \hh +\hh \tfrac{1}{N}\xx_\n )
\end{equation*}
\resizebox{\columnwidth}{!}{for $\xx_\n\hh\in\hh \X_\n $ and $\byyag_{-\n}\hh \in \hh \Sxag_{-\n} \eqd \{\frac{1}{N}\sum_{\m \neq n }\xx_\m :\xx_\m \hh \in \hh \X_\m, \forall \m \}$.%
}

\smallskip

\noindent In this paper, we adopt the following standard assumptions:
\newpage
\begin{assumption}\label{assp_convex_costs}~~ \nopagebreak
\begin{enumerate}[wide,labelindent=4pt,label=\roman*)]
\item for each player $\n \in\N $, the set $\X_\n $ is a convex and compact subset of $\rit^T$ with nonempty relative interior. Besides, there exists a fixed $\diamX>0$, such that $\max_{\n \in\N}\max_{\xx\in \X_n} \norm{\xx_n} \leq \diamX$;

\item for each $n\in\N$, the cost function $f_n$ and modified cost function $\hf_n$ are convex in the first variable;
\item for each $n\in\N$, the cost function $f_n:(\xx_n,\bxxag) \mapsto f_n(\xx_n,\bxxag)$ is Lipschitz continuous in $\bxxag$, with Lipschitz constant $L_2$ independent of $n$ and $\xx_n$; \label{assp:it:costLip}
\item $\Aset$ is a convex closed set of $\rit^T$, and $\Sxag\cap \Aset$ is not empty.

\end{enumerate}
\end{assumption}

In particular, \Cref{assp_convex_costs}.i) implies that  $\Sxag$ has a nonempty relative interior.
Let us now consider the subgradient correspondences $\corProdGradHf: \FX \rightrightarrows \rit^{NT}$ and $\corProdGradf: \FX \rightrightarrows \rit^{NT}$, defined for any $\xx\in \FX$ as:
\begin{align*}
\corProdGradHf(\xx) &\eqd \hh \{(\bm{h}_\n )_{\n\in \N }\hh \in\hh \rit^{NT}: \bm{h}_\n  \hh \in \hh \partial_1 \hf_\n (\xx_\n , \bxxag_{-\n }) , \  \forall \n \in \N \} \\
&= \txt\prod_{\n\in\N } \partial_1 \hf_\n (\xx_\n , \bxxag_{-\n }) \  ; \\
\corProdGradf(\xx) & \eqd\{(\bm{h}_\n )_{\n\in \N }\in \rit^{NT}: \bm{h}_\n  \in \partial_1 f_\n (\xx_\n , \bxxag) , \  \forall \n \in \N \} 
\\ &= \txt\prod_{\n\in\N } \partial_1 f_\n (\xx_\n , \bxxag)\  ,
\end{align*}
where $\partial_1$ denotes the partial differential w.r.t. the first variable of the function, while $\prod_\n$ denotes here the Cartesian product.
The interpretation of $\corProdGradHf(\xx)$ is clear: $\bm{h}_\n $ is a subgradient of player $ \n $'s utility function $\hf_\n $ w.r.t. her action $\xx_\n $. Let us leave the interpretation of $\corProdGradf(\xx)$ until \Cref{def:pseudoVNE}. 
For the moment, let us establish a relation between $\corProdGradHf$ and $\corProdGradf$:
\begin{lemma}\label{prop:subgradients-sets}
For each $\xx \in \FX$, $\bm{h} \in \corProdGradHf(\xx)$ if and only if there exists $\g_1 \in  \txt\prod_{\n\in\N } \partial_1 f_\n (\xx_\n , \bxxag) = \corProdGradf(\xx)$ and $\g_2 \in  \txt\prod_{\n\in\N } \partial_2 f_\n (\xx_\n , \bxxag) $ such that $\bm{h}= \g_1 + \tfrac{1}{N}\g_2$.
\end{lemma}
\begin{proof}
The proof can be obtained from \cite[Prop. 16.6 and 16.7]{combettes2011monotone}. Details are omitted for brevity.
\end{proof}
\smallskip

In the presence of coupling constraints, the notion of Nash Equilibrium (NE) \cite{nash1950equilibrium} is replaced by the one of Generalized Nash  Equilibrium (GNE). 
A profile $\xx\in \FX(\Aset)$ is a GNE if, for each player $ \n $:
\begin{equation*}
\resizebox{\columnwidth}{!}{$\hf_\n (\xx_\n , \bxxag_{-\n}) \hh \leq  \hf_\n (\yy_\n , \bxxag_{-\n}), \ \forall \yy_\n \hh \in\hh \X_n \text{ s.t. }(\textstyle \tfrac{\yy_\n}{N} \hh+ \hh \bxxag_{-\n})\hh \in \Sxag\cap \Aset.$}
\end{equation*} 
 For atomic games, Variational Nash Equilibria (VNE) \cite{harker1991generalized,Kulkarni2012vne} form a special class of GNEs with  symmetric properties: in some sense, the burden of constraint $\bxxag \in \Aset$ is shared symmetrically by players \cite{harker1991generalized}.
  VNEs, in the subdifferentiable case, are characterized as the solution of   a generalized VI (GVI) \eqref{cond:ind_opt_ve} stated below. 

\begin{definition}[Variational Nash Equilibrium (VNE), \cite{harker1991generalized}]\label{def:ve-finite}
A VNE   is a solution $\hxx\in \FX(\Aset)$ to the following GVI problem: 
 \begin{align}\label{cond:ind_opt_ve}
\exists\, \g\in \corProdGradHf(\hxx)  \text{ s.t. } %
&  \txt \big\langle \g, \xx- \hxx\big\rangle\geq 0,\; \forall \xx \in \FX(\Aset).
 \end{align}
\end{definition}

In particular, if  $\Sxag\subset \Aset$,  a VNE is an NE. We refer to \cite{facchinei2010gnep} for further discussions on VNE and VI characterization. 
In this paper, we adopt the notion of VNE  as the equilibrium concept in the presence of aggregate constraints. 

As the  first step of approximation, we define a \emph{nonatomic} aggregative game $\GA'$ associated to $\GA$, where each player $ \n $ is replaced by a continuum of identical nonatomic players. The set of nonatomic players is represented by the  interval $[0,1]$.
 Each player in population $ \n $ has action set $\X_\n $ and cost function $f_n$.  
 In nonatomic games, the concept of VNE is replaced by the concept of Variational Wardrop Equilibrium (VWE). We refer to \cite{PaulinWan2018nonsmooth} for the formal definition of a VWE and further discussions.
 In general, a VWE is characterized  by an infinite dimensional VI. Here, we only consider \emph{symmetric} VWE, characterized by a VI of finite dimension as defined below:
\begin{definition}\label{def:pseudoVNE}
A \emph{Symmetric Variational Wardrop Equilibrium} (SVWE) of $\GA'$  is  a solution $\sxx\in \FX(\Aset)$ to the GVI:
\begin{align}\label{eq:def-pseudoGVI}
  \exists \g\in \corProdGradf(\sxx) \text{ s.t. } %
 & \langle \g, \xx-\sxx \rangle\geq 0,\; \forall \xx\in \FX(\Aset)\ .
 \end{align}
 where for each $\n\in\N$, $\sxx_n$ is the common action adopted by all nonatomic players in population $\n$.
We denote by $\bsxxag \eqd \sum_n \sxx_n$ the SVWE average action.
\end{definition}

In addition to the population  game interpretation, an SVWE can be interpreted as an equilibrium concept in the initial game, where the contribution of each atomic player on the average action $\bxxag$ is negligible \cite{gentile2017nash}.
Given an SVWE $\sxx$, the profile $\sxx_n$ of a player $n$  will not minimize her cost function $\hat{f}_n(\cdot,\bsxxag_{-n})$, but will minimize $f_n(\cdot,\bsxxag)$ considering $\bsxxag$ as fixed. %

The existence of an SVWE in the population game $\GA'$, with players identical in the same population, is given in \Cref{prop:exist_ve}. %

\begin{proposition}[Existence of equilibria]\label{prop:exist_ve}
  Under \Cref{assp_convex_costs}, $\GA$ (resp. $\GA'$) admits a VNE (resp. SVWE). 
\end{proposition}
\begin{proof} From \cite[Prop. 8.7]{rockafellar2009variational}, we obtain that  
$\corProdGradHf$ and $\corProdGradf$ are nonempty, convex, compact valued, upper hemicontinuous correspondences.
Then \cite[Cor. 3.1]{chanpang1982gqvip} shows that the GVI problems \eqref{cond:ind_opt_ve} and \eqref{eq:def-pseudoGVI}  admit a solution on the finite dimensional convex compact $\FX(\Aset)$. %
\end{proof}

Before discussing the uniqueness of equilibria, 
 let us recall some relevant definition of monotonicity for correspondences:
\begin{definition}\label{def:mono_atom}
A correspondence $\Gamma:\FX \rightrightarrows \rit^{NT}$ is said:
\begin{enumerate}[wide,labelindent=4pt,label=\roman*)]%
\item\emph{monotone} if for all $\xx, \yy \in \FX , \g\in \Gamma(\xx), \bh \in \Gamma(\yy)$:\label{cd:mono_atom}
\begin{equation*}
\txt\sum_{\n\in \N } \langle \g_\n  - \bh_\n , \xx_\n  - \yy_\n  \rangle \geq 0\ ;
\end{equation*}

\item \emph{strictly monotone} if the equality  holds \textit{iff} $\xx=\yy$;

\item \resizebox{0.93\columnwidth}{!}{\emph{aggregatively strictly monotone} if  equality  holds \textit{iff} $\underset{n}{\sum} \xx_\n\hh = \hh\underset{n}{\sum} \yy_\n $;}
\item $\stgccvut$-\emph{strongly monotone} if $\stgccvut>0$ and, for all $\xx, \yy \in \FX$:
\begin{equation*}
\txt\sum_{\n} \langle  \g_\n  -\bh_\n , \xx_\n  - \yy_\n  \rangle \hh\geq\hh \stgccvut\|\xx-\yy\|^2, \, \forall\g\!\in\! \Gamma(\xx), \bh\! \in\! \Gamma(\yy)\, ;
\end{equation*}

\item $\beta$-\emph{aggregatively strongly monotone} on $\FX$ if $\beta>0$ and, for all $\xx, \yy \in \FX$ with $\bxxag=\tfrac{1}{N}\sum_\n \xx_\n $, $\byyag=\tfrac{1}{N}\sum_\n \yy_\n $: 
\begin{equation*}
\txt \sum_{\n } \langle \g_\n  - \bh_\n , \xx_\n  - \yy_\n  \rangle\! \geq\! N\beta\|\bxxag -\byyag\|^2 , \, \forall\g\!\in\! \Gamma(\xx), \bh\! \in\! \Gamma(\yy) .
\end{equation*}
\end{enumerate}

\end{definition}

\smallskip

If $T=1$,  ``monotone" corresponds to ``increasing". 
Besides, (aggregatively) strict monotonicity implies monotonicity, while strong (resp. aggregatively strong) monotonicity implies strict (resp. aggregatively strict) monotonicity.
\smallskip

\Cref{th:unique_vnevwe} recalls some results on the uniqueness of VNE and SVWE, according to the  monotonicity of  $\corProdGradHf$ and $\corProdGradf$:
\begin{proposition}[Uniqueness of equilibria]\label{th:unique_vnevwe}%
Under \Cref{assp_convex_costs}: 
\begin{enumerate}[wide,label=\roman*),labelindent=4pt]
\item if $\corProdGradHf$ (resp. $\corProdGradf$) is strictly monotone, then $\GA$ (resp. $\GA'$) has a unique VNE (resp. SVWE); 
\item if  $\corProdGradHf$ (resp. $\corProdGradf$) is aggregatively strictly monotone, then all VNE (resp. SVWE) of $\GA$ (resp. $\GA'$) have the same aggregate profile.
\end{enumerate}
\end{proposition}
\proof The proof can be adapted from \cite[Chapter 2]{facchinei2007finite}.
\smallskip

Except in some particular cases (e.g. $\cc$ linear \cite{orda1993competitive,richman2007topounique}), the operator $\corProdGradHf$ is rarely monotone.
The  monotonicity of $\corProdGradf$, however, can usually be obtained under simple assumptions:  for a splittable congestion game  (\Cref{def:splittableCongGame}), as stated in 
 \cite[Lemma 3]{gentile2017nash}, we have:
 \begin{enumerate}[wide,label=\roman*),labelindent=0pt]
   \item if for each $t\in\T$, $c_t$ is convex and nondecreasing, and for each $\n\in\N$, $u_n$ is concave, then  $\corProdGradf$ is monotone;
\item if in addition to i), for each  $ \n  \in \N $, $u_\n $ is $\stgccvut_\n $-strongly concave,  then $\corProdGradf$ is $\stgccvut$-strongly monotone with $\stgccvut\eqd \min_{\n \in\N }\stgccvut_\n $;
\item if in addition to i), for each $t\in\T$,  $c_t$ is $\beta_t$-strictly increasing, then $\corProdGradf$ is $\beta$-aggregatively strongly monotone with  $\beta\eqd \min_{t\in\T}\beta_t$.
\end{enumerate}

In view of \Cref{th:unique_vnevwe}, the absence of monotonicity of $\corProdGradHf$ can result in multiple VNEs \cite{bhaskar2009notunique}. 
However, the next section shows that, with a large number of players, those VNEs 
 are close to each other, and are well approximated by an SVWE.%

\section{Approximating VNEs of a Large Game}
\label{sec:approxRes}
\subsection{Considering SVWE instead of VNE}
\label{subsec:SVWEapproxVNE}

Let $\X_0\subset \rit^T$ be the closed convex hull of $\bigcup_{\n\in \N }\X_\n $. Then, we have from \Cref{assp_convex_costs}.i): 
 \begin{equation*}
 \max_{\xx \in \X_0} \norm{\xx} \leq \diamX \ .
\end{equation*}

\Cref{th:pseudoVNE} gives the first step of approximation of VNEs, by bounding the distance between a VNE and an SVWE.%
\begin{theorem}[SVWE close to VNE]\label{th:pseudoVNE}
Under \Cref{assp_convex_costs}, let $\xx \hh\in \hh \FX(\Aset)$ be a VNE of  $\GA$ and $\sxx\in \FX(\Aset)$ an SVWE of $\GA'$: %
\begin{enumerate}[wide,label=\roman*),labelindent=4pt]
\item if $ \corProdGradf$ is $\stgccvut$-strongly monotone, then $\xx^*$ is unique and we have: 
\begin{equation*}
  \|\xx-\sxx\|  \leq \frac{L_2}{\stgccvut}\frac{1}{\sqrt{N}}. %
\end{equation*}
From this bound and  Cauchy-Schwartz inequality, we also have:
\begin{equation*}
 \tfrac{1}{N}\sum_\n\|\xx_\n -\sxx_\n\|  \leq \frac{L_2}{\stgccvut N}  \    \text{ and } \  
    \|\bxxag-\bsxxag\|  \leq  \frac{L_2}{\stgccvut N}  , %
  \end{equation*}
  where $L_2$ is the Lipschitz constant defined in \Cref{assp_convex_costs};
  \item if  $ \corProdGradf$ is $\beta$-aggr. str. monotone, then $\bxxag^*$ is unique and: 
\begin{equation*} %
\|\bxxag-\bsxxag\|\leq   \sqrt{\tfrac{ 2 \diamX L_2 }{\beta N }} . %
\end{equation*}
 \end{enumerate}
\end{theorem}

\smallskip
\ifonlineversion
\begin{proof}

 i) We start by observing that for any $\yy \in \FX$, if $\bm{g}_2 \in \prod_n \partial f_n(\yy_n, \byyag)$ then, from \Cref{assp_convex_costs}.iii) we get
$ \norm{\bm{g}_2} \leq \sqrt{N} L_2 $  .

  By definition of $\xx$ VNE (resp. $\sxx$), there exists $\bm{h} \in \corProdGradHf(\xx)$ (resp. $\bm{h}^* \in  \corProdGradf(\xx^*)$) such that:
  \begin{align*}
 \big\langle \bm{h}, \xx- \sxx\big\rangle\leq 0 \text{ and }  \big\langle \bm{h^*}, \sxx- \xx\big\rangle\leq 0 \ .
  \end{align*}
  Adding those two inequalities, we obtain:
  \begin{equation*}
     \big\langle \bm{h}^* -\bm{h}, \sxx - \xx \big\rangle\leq 0 \ .
  \end{equation*}
  Besides, from \Cref{prop:subgradients-sets}, there exists $\g_1 \in  \corProdGradf(\xx)$ and $\g_2 \in  \prod_n \partial f_n(\xx_n, \bxxag)$ such that $\bm{h} = \g_1 + \tfrac{1}{N} \g_2$.
  Then, using the $\stgccvut$-strong monotonicity of  $\corProdGradf$, we have:
  \begin{align*}
    \stgccvut \norm{\xx^* - \xx }^2  & \leq      \big\langle \bm{h}^* -\g_1, \sxx - \xx \big\rangle \\
  &  =       \big\langle \bm{h}^* -\bm{h}, \sxx - \xx \big\rangle  + \big\langle  \tfrac{1}{N} \g_2, \sxx - \xx \big\rangle \\
 &   \leq 0 +  \big\langle  \tfrac{1}{N} \g_2, \sxx - \xx \big\rangle \\
&    \leq  \tfrac{1}{N} \norm{\g_2} \norm{\sxx - \xx } \leq  \tfrac{L_2}{\sqrt{N}} \norm{\g_2} \norm{\sxx - \xx } ,
  \end{align*}
  using Cauchy-Schwarz inequality.
  We conclude by  simplifying by $  \norm{\sxx - \xx }$ to obtain the first bound of  \Cref{th:pseudoVNE}.i).
    
Besides, by Cauchy-Schwarz inequality, we have: 
\begin{equation*}
(\sum_\n\|\xx_\n -\sxx_\n\| )^2 \leq N \sum_\n\|\xx_\n -\sxx_\n\|^2 \leq \tfrac{N L_2^2}{\stgccvut^2N^2}  =\tfrac{ L_2^2}{\stgccvut^2N}  
\end{equation*}
which gives the second bound  of  \Cref{th:pseudoVNE}.i). 
Similarly, we have: 
\begin{equation*}
N\|\bxxag - \bsxxag\|=\|\sum_\n (\xx_\n - \sxx_\n)\| \leq \sqrt{ N \|\xx - \sxx \|^2} \leq \tfrac{L_2}{\stgccvut} .
\end{equation*}

\smallskip

ii) The proof is similar, except that we exploit the  $\beta$-strong monotonicity of  $\corProdGradf$ to get:
\begin{align*}
N\beta  \|\bxxag-\bsxxag\|^2 &\leq      \big\langle \bm{h}^* -\g_1, \sxx - \xx \big\rangle \\
       &    \leq  \tfrac{1}{N} \norm{\g_2} \norm{\sxx - \xx } \\
  &\leq   \tfrac{1}{\sqrt{N}} L_2  \times 2 \sqrt{N} \diamX  \ ,
\end{align*}
which leads to the desired bound. 

  \else
  For \textit{i)}, one can adapt \cite[Th.~1.2)]{gentile2017nash} to the subdifferential case by relying on \Cref{prop:subgradients-sets}.
  Item \textit{ii)} generalizes \cite[Th.~1.3]{gentile2017nash} for   \emph{aggr. monotone} games. We refer to online version \cite{PaulinWan2019estimationOnline} for details.
  \fi
\end{proof}  

While VNEs are not unique in general (see \Cref{th:unique_vnevwe}), by  applying the triangle inequality to the results of \Cref{th:pseudoVNE}  we observe that all VNEs are close to each other when $N$ is large. 

\Cref{th:pseudoVNE} shows that an (average) SVWE approximates an (average) VNE of $\GA$ when the number of players is large.
However, this does not reduce the dimension of the GVI to resolve, as the GVI characterizing the VNE and those characterizing  SVWE have the same dimension. 
For this reason, the second step of approximation consists of  regrouping the similar populations.%

\subsection{Classification of populations}
\label{subsec:class}

In this subsection, we shall %
 regroup the populations of the game  $\GA'$ having similar strategy sets $\X_\n $ and utility subgradients $\partial(-u_\n )$,  into  larger populations, endow them with a common strategy set and a common utility function, so that the SVWE of this new population game approximates the SVWE of $\GA'$.
 The similarity between two sets $\X$ and $\Y$ is measured through the Hausdorff distance:
 \begin{equation*}
 d_H(\X,\Y) \eqd \max\big(  \max_{\xx\in\X} d(\xx,\Y), \ \max_{\yy\in\Y} d(\yy,\X) \big) \ , 
\end{equation*}
where $d(\xx,\Y)\eqd \inf_{\yy \in\Y} \norm{\xx-\yy}$ is the standard distance function.
Let us define the  compact set %
$ \M \eqd \X_0 +B_0(\delta) \ , $
 where $B_0(\delta)$ is the ball centered at 0 and of radius $\delta$, and $\delta\geq \max_{n,m \in \N} d_H(\X_n,\X_m) $.
W.l.o.g, we assume that for each $ \n \in \N $, $f_\n $ can be extended to a neighborhood of $\M^2$, and is bounded on $\M^2$. 
We make the additional assumption: 
\begin{assumption} \label{assp:lipschitzVar1}
For each  $n\in\N$ and each $\bxxag \in \Sxag$, $f_n(\cdot,\bxxag)$ is Lipschitz continuous with constant $L_{1n}$, and $L_1\hh \eqd \hh \max_n L_{1n}\hh < \hh \infty$.
\end{assumption}

\smallskip

At the SVWE of the game with reduced dimension, all the nonatomic players in the same population play the same action, by definition of an SVWE.
 Therefore, in order for this new SVWE to well approximate the SVWE in $\GA'$, a condition is that similar populations in $\GA'$ do play similar actions at the SVWE of $\GA'$: this is obtained easily in the case  without coupling constraint, as shown in 
\Cref{prop:continuiteNE}. %

\begin{proposition}\label{prop:continuiteNE}
Under \Cref{assp_convex_costs}, let $\sxx\in \FX$ be an SVWE of $\G'(\Aset)$ with $\Aset =\rit^T$ (no coupling constraint).  
For  populations $ \n ,\m \in \N $, if  $f_\n(\cdot,\bsxxag) $ is $\stgccvut_\n $-strongly convex, $d_H(\X_\n , \X_{\m }) \leq \dset$,  and if:
\vspace{-0.2cm}
\begin{equation*} \sup_{\g_\m \in \partial_1 f_m(\sxx_m, \bsxxag)} d(\g_\m , \partial_1 f_n(\sxx_m, \bsxxag) ))\leq \duti_*
\end{equation*}
 then: 
\begin{equation*}
\norm{\sxx_\n  - \sxx_\m }^2\leq \tfrac{1}{\stgccvut_\n }\big((L_{1n}+ L_{1m})\dset  +2 \duti_*   \diamX\big).
\end{equation*}
\end{proposition}

\smallskip

\ifonlineversion

\begin{proof}
  Let $\g(\sxx)\in \corProdGradf(\sxx)$ satisfying the GVI \eqref{eq:def-pseudoGVI}. As there is no coupling constraint, for any $\xx_\n \hh \in \X_\n$ (resp. $\xx_\m \hh \in \X_\m$), one  can consider $\xx \eqd (\xx_n,\sxx_{-n}) \in \FX$ (resp. $\xx \eqd (\xx_m,\sxx_{-m})$) in \eqref{def:pseudoVNE} to obtain:
  \begin{equation*}
 \langle \g_\n (\sxx ), \sxx_\n  -\xx_\n  \rangle \leq 0 \text{ and }      \  \langle \g_\m (\sxx ), \sxx_\m  -\xx_\m  \rangle \leq 0 
  \end{equation*}
Let $\bh_\n (\sxx ) \in \partial_1 f_n(\sxx_\m, \bsxxag )$ be such that $\|\bh_\n  (\sxx )- \g_\m  (\sxx )\|\leq \duti_*$. 
Then, by the strong convexity of $f_n(\cdot,\bsxxag) $:
\begin{align*}
&\stgccvut_\n  \norm{\sxx_\n   -\sxx_\m }^2  \leq   \langle \g_\n (\sxx )-\bh_\n (\sxx ), \sxx_\n -\sxx_\m  \rangle \\
= & \langle \g_\n (\sxx )-\g_\m (\sxx )+\g_\m (\sxx )-\bh_\n (\sxx ), \sxx_\n -\sxx_\m  \rangle   \\
\leq & \langle \g_\n (\sxx ) - \g_\m (\sxx ),\sxx_\n -\sxx_\m   \rangle + 2  \duti_* \diamX \ .
\end{align*}
Making use of the inequalities \eqref{eq:def-pseudoGVI} on $\g_\n(\sxx)$  (resp. $\g_\m(\sxx)$), we obtain that the first part of this last quantity is bounded as:
\begin{align*}
 &\langle \g_\n (\sxx ) - \g_\m (\sxx ),\sxx_\n - \Proj_{\n}(\sxx ) +\Proj_{\n}(\sxx )- \sxx_\m    \rangle  \\
\leq & 0+   \langle \g_\n (\sxx )  ,\Proj_{\n}(\sxx )-\sxx_\m   \rangle   
 +\langle \g_\m (\sxx )  ,\Proj_{\m }(\sxx_\n )- \sxx_\n  \rangle\\
\leq & L_{1n}\dset + L_{1m}\dset \ ,
\end{align*}
where $\Proj_\n $ (resp. $\Proj_\m $) is the projector on $\X_\n $ (resp. $\X_\m $).
\end{proof} 
\fi

In the presence of coupling constraints,  a similar result is harder to obtain: because there is a coupling in the admissible deviations $(\xx_n)_n \in \FXA$, the  SVWE condition \eqref{def:pseudoVNE} cannot be translated into individual conditions.

Let us now present the regrouping procedure.
Denote by  $\tG\epopset(\Aset)$  an auxiliary game, with a set $\popset$ of $\popcard$ populations.
 Each population $\pidx \in \popset$ corresponds to a subset $\N_\pidx$ of populations in the  game $\GA'$, such that $\bigcup_{\pidx\in \popset}\N_\pidx = \N $ and for any $\pidx,j\in\popset, \N_\pidx \cap \N_j =\emptyset$, i.e $(\N_\pidx)_{\pidx\in\popset}$ forms a partition of $\N$.
 Denote $\popcard_\pidx = |\N_\pidx|$ the number of original populations now included in $\pidx$.
 Each nonatomic player in population $\pidx$ is represented by a point $\theta\in [0,\popcard _\pidx]$. 
 The common action set of each nonatomic player in $\pidx$ is a compact convex subset of $\rit^T$, denoted by $\X_\pidx$. 
 Each player $\th$ in each population $\pidx$ playing action $\xx_\th$, let 
\begin{equation*}\bxxag \eqd \tfrac{1}{N}\sum_{\pidx \in \popset}\int_{\th\in [0,\popcard_\pidx]} \xx_\th \dth
\end{equation*} denotes the aggregate action profile. The set of aggregate action-profile is a subset of $\rit^T$ given as:
\begin{equation*}
\Sxag\epopset \eqd \{\tfrac{1}{N}\txt \sum_{\pidx\in \popset}\int_{\th\in [0,\popcard_\pidx]} \xx_\th \dth: \xx_\th\in \X_\pidx, \, \forall \th\in [0,\popcard_\pidx], \, \forall \pidx\in \popset\} \  . \end{equation*}
All players $\th$ in population $\pidx \in \popset$ have a similar cost function denoted by $f_\pidx: (\xx_{\th},\bxxag) \mapsto f_\pidx(\xx_\th, \bxxag)$, convex w.r.t. $\xx_\th$.

We are only interested in \emph{symmetric} action profiles, where all the nonatomic players in each population $ \pidx $ play the same action. 
 Denote the set of symmetric action profiles by $\FX\epopset=\prod_{\pidx\in \popset} \X_\pidx$. 
 Considering the coupling constraint  $ \Aset$, let: \begin{equation*}
 \FX\epopset(\Aset)\eqd \{\xx\in\FX\epopset:\bxxag=\tfrac{1}{N} \txt \sum_{\pidx\in \popset}\popcard _\pidx\xx_\pidx\in \Aset \}\ .
 \end{equation*}

\noindent We introduce two indicators to quantify the clustering of $\tG\epopset$:
\begin{itemize}[wide]
\item $\mdset = \max_{\pidx\in \popset} \dset_\pidx$, where
\begin{equation} \label{eq:def_dsetAtom}
\dset_\pidx \eqd \txt\max_{\n  \in\N_\pidx} d_{H}\left( \X_\n , \X_\pidx \right)\ ,
\end{equation} 
\item $\mduti = \max_{\pidx\in \popset} \duti_\pidx$, where
\begin{equation}\label{eq:def_dutAtom}
\duti_\pidx \eqd  \max_{\n  \in \N_\pidx}\sup_{\xx_\pidx\in \X_\pidx, \bxxag \in \M}   \hh\hh d_H\left(  \partial_1 f_i(\xx_\pidx,\bxxag)  , \partial_1 f_n(\xx_\pidx,\bxxag)\right).
\end{equation}
\end{itemize}
The quantity $\dset_\pidx$ measures the heterogeneity in strategy sets of populations within the group $\N_\pidx$, while  $\duti_\pidx$ measures the heterogeneity in the subgradients in the group $\N_\pidx$.

Since the auxiliary game $\tG\epopset$ is to be used to compute an approximation of a VNE of $\G$, the indicators $\dset_\pidx$ and $\duti_\pidx$ should be minimized when defining $\tG\epopset$. 
Thus, we assume that $(\X_\pidx)_\pidx$ and $(u_\pidx)_\pidx$ are chosen such that the following holds:
\begin{assumption} \label{assp:Xnandun} For each $\pidx \in\popset$, we have:
\begin{enumerate}[wide,label=\roman*),labelindent=4pt]
\item $\X_\pidx$ is a subset of  $\X_0$, the convex hull of $\bigcup_{\n\in \N_\pidx}\X_\n $, so that $ \max_{\xx\in \X_\pidx}  \| \xx \| \leq \diamX$. %
Moreover, for each $ \n \in\N_\pidx$, $\aff \X_\n  \subset \aff \X_\pidx$, where $\aff S$ denotes the affine hull of $S$;
\item similarly,  $f_\pidx$ is chosen such that $\partial_1 f_i(\xx_\pidx,\bxxag)$ is contained in the convex hull of $\bigcup_{\n\in \N_\pidx}\partial_1 f_n(\xx_\pidx,\bxxag)$ for all $\xx\in \X_\pidx$ and $\bxxag \in \M$, so that $\sup_{\g_i\in \partial_1 f_i(\xx_\pidx,\bxxag)}\norm{\g_i} \leq \displaystyle\max_{\n\in \N_\pidx}L_{1n} \leq L_1$.
\end{enumerate}
\end{assumption}

An interesting case appears  when $\N $ can be divided into homogeneous populations, as in \Cref{ex:homogenousPops} below:

\begin{example} \label{ex:homogenousPops}
The player set $\N $ can be divided into a small number $\popcard$  of subsets $(\N_\pidx)_{1\leq \pidx \leq \popcard}$, with homogeneous players inside each subset $\N_\pidx$ (i.e., for each $\pidx$ and $ \n ,\m  \in \N_\pidx$, $\X_\n =\X_\m  $ and $f_\n =f_\m $). In that case, consider an auxiliary game $\tG\epopset$ with $\pidx$ populations and, for each $\pidx \in \popset$ and $ \n \in\N_\pidx$, $\X_\pidx \eqd \X_\n $ and $f_\pidx \eqd f_\n $. Then, $\mdset=\mduti=0$.
\end{example}

In order to approximate the SVWE of $\GA'$ by the SVWE of an auxiliary game $\tG\epopset$, let us first state the following result on  the geometry of the action sets for technical use.
\begin{lemma}\label{lm:intprofileAt}
Under \Cref{assp_convex_costs}, there exists a strictly positive constant $\rho$  and an action profile $\zz\in \FX(\Aset)$ such that, $d(\zz_\n , \rbd \X_\n )\geq \rho$ for all $\n \in \N $, where $\rbd$ stands for the relative boundary.
\end{lemma}
\begin{proof} See \Cref{app:proof:lm-intprofile}. \end{proof}

Given a \emph{symmetric} action profile $\xx\epopset \in \FX\epopset(\Aset)$ in the auxiliary game $\tG\epopset(\Aset)$, we can define a corresponding symmetric action profile of $\GA'$ such that all the nonatomic players in the populations regrouped in $\N_\pidx$ play the same action $\xx\epopset_\pidx$(it is allowed that  $\xx\epopset_\pidx$ be not in $\X_\n $). 
 Formally, we define the map $\psi: \rit^{\popcard T} \rightarrow \rit^{N T}$:
\begin{equation*}
\forall \xx\epopset \hh\hh \in\hh \rit^{\popcard T},\, \psi(\xx\epopset) \hh=\hh(\xx_\n )_{\n\in \N } \, \text{ where } \xx_\n  \hh= \hh\xx\epopset_\pidx \ , \ \forall \n  \in \N_\pidx \, .
\end{equation*}
 Conversely, for a symmetric action profile $\xx$ in $\GA'$, we define a corresponding symmetric action profile in the auxiliary game $\tG\epopset(\Aset)$ by the following map $\bpsi: \rit^{NT} \rightarrow \rit^{\popcard T}$:
\begin{align*}
 \forall \xx \hh\in\hh \rit^{NT}, \,\bpsi(\xx)\hh=\hh(\xx\epopset_\pidx)_{\pidx\in \popset}\, \text{ where } \xx\epopset_\pidx \hh=\hh \tfrac{1}{\popcard _\pidx}\txt\sum_{\n\in\N_\pidx } \xx_\n .
\end{align*}

\Cref{th:estimNE:main} below is the main result of this subsection. It gives an upper bound on the distance between the SVWE of the population game $\GA'$---of same dimension as the original atomic game $\GA$---and the SVWE of an auxiliary game $\tG\epopset(\Aset)$ with reduced dimension.

\begin{theorem}[SVWE of $\tG\epopset(\Aset)$ is close to SVWE of $\G(\Aset)'$]\label{th:estimNE:main}%
Under \Cref{assp_convex_costs,assp:Xnandun}, in an auxiliary game $\tG\epopset(\Aset)$, $\mdset$ and $\mduti$ are defined by \eqref{eq:def_dsetAtom} and \eqref{eq:def_dutAtom}, with $\mdset <\frac{\rho}{2}$. Let $\xx$ be an SVWE of $\tG\epopset(\Aset)$, and $\sxx$ an SVWE of $\G(\Aset)'$. Then: 
\begin{enumerate}[wide,label=\roman*),labelindent=0pt]
\item if $ \corProdGradf$ is $\stgccvut$-strongly monotone, then $\xx$ and $\sxx$ are unique and
 \begin{align} \label{eq:boundClusterIndiv}
   & \|\psi(\xx)-\sxx\|  \leq \sqrt{ {N} \tfrac{ K\left(\mdset,\mduti\right)}{ \stgccvut} }\,
 \end{align}
 From this bound and  Cauchy-Schwartz inequality, we also have:
 \begin{align}
&  \tfrac{1}{N} \hh \sum_\n\|\psi_\n(\xx) \hh-\hh \sxx_\n\|  \leq   \hh \sqrt{\tfrac{ \hh K\left(\mdset,\mduti\right) }{  \stgccvut}  }  
 \text{ and }  
\left\| \bxxag \hh -\hh \bsxxag\right\|  \leq    \sqrt{\tfrac{ K\left(\mdset,\mduti\right) }{  \stgccvut}  };   \nonumber
\end{align}
\item if $ \corProdGradf$ is $\beta$-aggr. str. monotone, then $\bxxag$, $\bsxxag$ are unique and
\begin{equation} 
\left\| \bxxag- \bsxxag \right\| \leq  \sqrt{\tfrac{ K\left(\mdset,\mduti\right) }{\beta}  }\ ,
\end{equation}
\end{enumerate}
where  $K(\mdset ,\mduti) \underset{\mdset,\mduti \rightarrow 0}{ \longrightarrow} 0 $  is a quantity defined as:  
\begin{equation} \label{eq:KboundExpression}
K(\mdset ,\mduti)\eqd 2 \diamX \big( 3 \tfrac{L_1 }{\rho}\mdset + \mduti\big)  \ \text{, with $L_1$ defined in  \ref{assp:lipschitzVar1}}.
\end{equation}
\end{theorem}
\smallskip

\begin{proof}  Let us start by two technical lemmas, for which the proofs are given in \Cref{app:proof:main}.
\begin{lemma}\label{lm:FYAt}~~%
\begin{enumerate}[wide,labelindent=4pt,label=\roman*)]
\item  $\forall \pidx \in \popset$,  $\forall \xx\in \X_\pidx$, if $d(\xx, \rbd \X_\pidx)>\dset_\pidx$, then $\xx\in \X_\n ,  \forall \n \in \N_\pidx$.
\item  $\forall \pidx \in \popset, \n \in \N_\pidx$,  $\forall \xx\in \X_\n $, if $d(\xx, \rbd \X_\n )>\dset_\pidx$, then $\xx\in \X_\pidx$.
\end{enumerate}
\end{lemma}
\begin{lemma}\label{lem:dist_agg_genized_sets_at}
Under \Cref{assp_convex_costs}, if $\mdset <\frac{\rho}{2}$, then:
\begin{enumerate}[wide,labelindent=4pt,label=\roman*)]
\item for  each $\xx \in \FX\epopset(\Aset)$, there is $\ww\in \FX(\Aset)$ such that $\|\ww_\n  - \psi_\n (\xx)\| \leq 4 \diamX \tfrac{\mdset}{\rho}$ for each $ \n \in \N $;
\item for each $\xx\in \FX(\Aset)$, there is $\ww\in \FX\epopset(\Aset)$ such that $\|\ww_\pidx-\bpsi_\pidx(\xx)\|\leq 2 \diamX \popcard_\pidx\frac{\mdset}{\rho}$ for each $n\in \popset$.
\end{enumerate}
\end{lemma}

Let $\ww\in \FX(\Aset)$ be s.t. $\forall \n \in \N $, $\|\ww_\n  - \psi_\n (\xx)\| \leq 4 \diamX \mdset/\rho$  (cf. \Cref{lem:dist_agg_genized_sets_at}). Since $\sxx$ (resp $\xx$) is an SVWE of $\GA'$ (resp. $\tG(\Aset)$), there exists $\g(\sxx) \in \corProdGradf(\sxx)$  (resp. $\big(\bh_\pidx(\xx)\big)_{\pidx\in\popset}\hh\hh \in\hh \prod_\pidx \partial_1 f_i(\xx_\pidx,\bxxag)$). 
\begin{equation*}
\txt\sum_{\n} \langle \g(\sxx)  , \ \sxx_\n   - \ww_\n   \rangle \leq 0 \ \text{ and } \txt \sum_{\pidx}\hh  \popcard_\pidx \langle \bh_\pidx(\xx), \xx_\pidx - \yy_\pidx\rangle \leq 0.
\end{equation*}  
for all $\yy \in \X\epopset(\Aset)$. 
Moreover, for each $ \pidx$ and 
$  \n  \in \N_\pidx$, by  definition of $\duti_\n $, let $\bm{r}_\n(\xx) \in \partial_1 f_n(\xx_\pidx,\bxxag)$ s.t. $ \|\bm{r}_\n(\xx) \hh - \hh\bh_\pidx(\xx)\|\leq \duti_\n $.
Then, $\langle \g(\sxx) - \bm{r}(\xx) ,  \sxx-  \psi( \xx) \rangle$ equals to:
\renewcommand{\-}{\hspace{-2pt}-\hspace{-2pt}}
\begin{equation*}
 \langle \g(\sxx)  ,  \sxx-  \ww \rangle + \langle \g(\sxx)  ,    \ww - \psi( \xx)\rangle    + \langle \bm{r}(\xx) ,    \psi( \xx) -\sxx \rangle.
 \end{equation*}
 The first term is nonpositive because $\sxx$ is SVWE. By definition of $\ww$, the second term is bounded by $ L_1 N 4 \diamX \mdset/\rho$. 
 Let us decompose the third term,  using
a  $\yy\in \FX\epopset(\Aset)$ s.t. $\forall \pidx$, $\|\yy_\pidx-\bpsi_\pidx(\sxx)\|_{\popset}\leq 2 \diamX \popcard_\pidx \mdset/\rho$ (cf. \Cref{lem:dist_agg_genized_sets_at}), into:
\begin{align*}
 \langle \bm{r}(\xx) \- \bh(\xx),  \psi( \xx) \- \sxx  \rangle  +\hh \langle \bh(\xx),  \psi( \xx) \- \yy  \rangle \hh +\hh \langle \bh(\xx),  \yy \-\sxx  \rangle.
 \end{align*}
The second term is nonpositive because $\xx$ is an SVWE of $\tG(\Aset)$. The first term is bounded by $2 NR \mduti$. The third and last term is bounded by $L_1 2 R N \mdset/\rho$. 
We conclude by using the strong monotonicity of $ \corProdGradf$, which gives for case i):
\begin{align*}
  \alpha \norm{\sxx-\psi(\xx) }^2 & \leq \langle \g(\sxx) - \bm{r}(\xx) ,  \sxx-  \psi( \xx) \rangle \\
                                  & \leq 0+   L_1 N 4 \diamX \mdset/\rho  + 2 NR \mduti + L_1 2 R N \mdset/\rho \ ,
\end{align*}
 and similarly, for case ii):
\begin{align*}
  \beta N \norm{\bsxxag-\bxxag }^2 & \leq \langle \g(\sxx) - \bm{r}(\xx) ,  \sxx-  \psi( \xx) \rangle \\
                                  & \leq 0+   L_1 N 4 \diamX \mdset/\rho  + 2 NR \mduti + L_1 2 R N \mdset/\rho\ .
\end{align*}
\end{proof}
\begin{remark}
The bound  \eqref{eq:boundClusterIndiv} given on the \emph{individual profiles} diverges with the number of players $N$.
 This is a consequence of the fact that individual  errors, $\| \xx_n - \xx_i^*\|$ for each $n$ within a population $\N_i$, may accumulate, and this accumulation  is captured by the Euclidean norm, as $\norm{ \mathds{1}_N}= \sqrt{N}$. 
\end{remark}
We have pointed out that the approximation error depends on how the populations are clustered according to $\popset$, and is related to the heterogeneity of players in $\N $ rather than their number. In particular, in the case of \Cref{ex:homogenousPops}, \Cref{th:estimNE:main} states that the (aggregate) SVWE of the auxiliary game $\tG\epopset(\Aset)$ is exactly equal to the (aggregate) SVWE of the large game $\G(\Aset)'$.

\subsection{Combining the two steps to approximate a VNE of $\GA$}

The following theorem is the main result of the paper, which is immediately obtained as the combination of the two steps of approximation given in \Cref{th:pseudoVNE} and in \Cref{th:estimNE:main} in the computation of a VNE of the original game $\GA$.
\begin{corollary}[SVWE of $\tG\epopset(\Aset)$ is close to VNEs  of $\G(\Aset)$]\label{th:bound}
Under \Cref{assp_convex_costs,assp:Xnandun},  in an auxiliary game $\tG\epopset(\Aset)$, $\mdset$ and $\mduti$ are defined by \Cref{eq:def_dsetAtom,eq:def_dutAtom}, with $\mdset <\frac{\rho}{2}$.
 Let $\sxx$ be an SVWE of $\tG\epopset(\Aset)$, $\xx\in \FX(\Aset)$ be a VNE of  $\GA$, $\bsxxag =\frac{1}{N} \sum_{\pidx\in \popset}\popcard _\pidx \sxx_\pidx$, $\bxxag=\frac{1}{N}\sum_{\n\in \N } \xx_\n $, and $K(\mdset,\mduti) $ the constant given by \eqref{eq:KboundExpression}.
\begin{enumerate}[wide,label=\roman*),labelindent=4pt]
\item if $\corProdGradf$ is $\stgccvut$-strongly monotone, then $\sxx$ is unique and
\vspace{-0.15cm}
 \begin{align*}
   \|\psi(\sxx)-\xx\| &\leq  \tfrac{L_2}{\stgccvut}\frac{1}{\sqrt{N}} + \sqrt{N \tfrac{ K\left(\mdset,\mduti\right) }{\stgccvut}  } \ .
 \end{align*}
 {From this  bound and  Cauchy-Schwartz inequality, we also have:}
  \begin{align*}
  \tfrac{1}{N}\sum_\n\|\psi_\n(\sxx)-\xx_\n\|  &\leq  \tfrac{L_2}{\stgccvut N} +\sqrt{\tfrac{ K\left(\mdset,\mduti\right) }{  \stgccvut}  }\, \\
 \text{ and }
  \left\| \bxxag- \bsxxag \right\| & \leq \tfrac{L_2}{\stgccvut N}+\sqrt{\tfrac{ K\left(\mdset,\mduti\right) }{  \stgccvut}  }\, ;
\end{align*}
\item  if $\corProdGradf$ is $\beta$-aggr. str. monotone,  then  $\bsxxag$%
 is unique,  and
 \vspace{-0.15cm}
\begin{equation*}
 \left\| \bxxag- \bsxxag \right\|  \leq \diamX \sqrt{ \tfrac{2 TC}{ N \beta } } +\sqrt{\tfrac{ K\left(\mdset,\mduti\right) }{ \beta}  }\, .
\end{equation*} 
\end{enumerate}
\end{corollary}

Given the large game $\G(\Aset)$ and a certain $\popcard \in \nit^*$, \Cref{th:bound} suggests that we should find the auxiliary game $\tG\epopset$ with $\I=\{1, \dots , \popcard \}$  that minimizes $K(\mdset,\mduti) $ in order to have the best possible approximation of the equilibria. This would correspond to a ``clustering problem'' given as follows:
\begin{equation} \label{eq:clusteringMinimizeK}
\min_{(\N_\pidx)_\pidx \in \mathcal{P}^\popcard(\N ) }\  \min_{(\X_\pidx)_{\pidx}} \ \min_{(u_\pidx)_{\pidx}} K(\mdset,\mduti) ,
\end{equation}
where $\mathcal{P}^\popcard(\N)$ denotes the set of all partitions of $\N $ of cardinal $\popcard$, while $(\X_\pidx)_{\pidx}$ and $(u_\pidx)_\pidx $ are chosen according to \Cref{assp:Xnandun}.%

The value of the optimal solutions of problem \eqref{eq:clusteringMinimizeK}, and thus of the quality of the approximation in \Cref{th:bound}, depends on the homogeneity of the $N$ players in $\N $ in terms of action sets and utility functions.
 The ``ideal'' case is given in \Cref{ex:homogenousPops} where $\N $ is composed of a small number $\popcard$ of homogeneous populations and thus $K(\mdset,\mduti)=0$.

Solving \eqref{eq:clusteringMinimizeK} is a hard problem in itself: it  generalizes the $k$-means clustering problem \cite{lloyd1982least}  (with $k=\popcard$ and considering a function of Hausdorff distances), which is  NP-hard \cite{garey1982complexity}.
 In  \Cref{sec:appliDRapprox}, we illustrate how we  use directly the $k$-means algorithm to compute efficiently an approximate solution $(\N_\pidx, \X_\pidx, u_\pidx)_{\pidx\in\popset}$ of \eqref{eq:clusteringMinimizeK} in the parametric case.

\smallskip
Finally, the number $\popcard$ in the definition of the auxiliary game should be chosen a priori  as a trade-off between the minimization of $K(\mdset,\mduti)$ and a sufficient minimization of the dimension. 
  Indeed, with $\popset = \N $, $\X_\n =\X_\pidx$ and $u_\pidx=u_\n $, we get $\mduti=\mdset=0$. 
However, the aim of \Cref{th:bound} is to find an auxiliary game $\tG\epopset$ with $\popcard \ll N$ so that the dimension of the GVIs characterizing the equilibria (and thus the time needed to compute their solutions) is significantly reduced, while ensuring a relatively small error, measured by $\mduti$ and $\mdset$.

\section{Application to Demand Response and Electricity Flexibilities}%
\label{sec:appliDRapprox}

Demand response (DR) \cite{ipakchi2009} refers to the techniques to optimize the electric consumption of distributed consumers %
 \cite{PaulinTSG17}. 
The increasing number of  electric vehicles (EV) offers a new source of flexibility in the optimization of the production and demand. %
To deal with privacy and distributed decisions and information issues, many works consider game theoretical approaches (considering consumers as players) for the management of consumption flexibilities \cite{saad2012game}.
Let us consider the charging of EVs on a set of 24-hour time-periods  $\T\eqd \{1,\dots,T\}$, with $T=24$, indexing the hours from 10 \textsc{pm} to $9$\textsc{pm} the day after (to include night time periods). %

\paragraph{Price functions: block rates energy prices}

The operator imposes electricity prices on each time-period, given as \emph{inclining block-rates tariffs} \cite{wang2017optimal}, 
i.e. a piece-wise affine function $c(\cdot)$ which depends on the aggregate demand  $X_t \eqd \sum_{\n} \x\nt$, where $\x\nt$ is the energy used by consumer $n$ during time period $t$ (one can easily transform it to a function of the average action $\overline{X}_t$ by scaling  coefficients):
\begin{equation}\label{eq:price}
\begin{split}
c(X)&= 1 + 0.1 X  \text{ if } X \leq 500 \\
  c(X)&=-49 + 0.2 X \text{ if } 500 \leq X \leq 1000 \\
  c(X) &=-349 + 0.5 X \text{ if } 1000 \leq X    \ .
\end{split}
\end{equation}
This function $c$ is continuous and convex. Prices are transmitted by the operator to each  consumer (EV owner) $ \n $ who minimizes, with respect to $\xx_n$, a cost function  of the \emph{congestion} form \eqref{eq:cost_player_def} with an energy cost determined by \eqref{eq:price} and a utility function $u_\n $ defined below. An equilibrium gives a stable situation where each consumer minimizes her cost and has no interest to deviate from her current profile. 

\paragraph{Consumers' constraints and parameters} \label{subsec:consumerParams}

We consider a set of $N=2000$ consumers who have demand constraints of the form:\begin{equation}
\label{eq:set_Edemand}
\hh \hh  \resizebox{0.91\columnwidth}{!}{$\X_\n \hh\eqd\hh \{ \xx_\n  \in \rit^T_+  :  \txt\sum_\t  \x\nt\hh =\hh E_\n  \text{ and }  \ux\nt \leq \x\nt \leq \ox\nt\}$}
\end{equation} 
where $ E_\n $ is the total energy needed by $ \n $, and $\ux\nt ,\ox\nt$  (physical) bounds on the power allowed at time $t$.  
The utility functions have the form 
$u_\n (\xx_\n ) \eqd - \omega_\n  \norm{\xx_\n -\yy_\n  }^{2} $, with $\yy_n$ a preferred charging profile.
Simulation parameters  are chosen as follows:
\begin{itemize}[wide]
\item $ E_\n $ is drawn uniformly between 1 and 30 kWh, which corresponds to a typical charge of a residential EV.
\item $\uux_\n , \oox_\n $:
 First, we generate, in two steps,  a set  of \emph{successive} charging time-periods $\T_\n =%
\{ h_\n  - \frac{\tau_\n }{2} , \dots , h_\n  + \frac{\tau_\n }{2} \}$; %
\begin{itemize}
\item the duration $\tau_\n $ is uniformly drawn from $ \{4,\dots,T \}$;
\item $h_\n $ is then uniformly drawn from $\{ 1+\frac{\tau_\n }{2} , \dots , T- \frac{\tau_\n }{2}\}$. 
\end{itemize}
Next, for $t\notin \T_\n $, let $\ux\nt=\ox\nt=0$ and, for $t\in \T_\n $, $\ux\nt$ (resp.  $\ox\nt$) is drawn uniformly from $[0,\frac{ E_\n }{\tau_\n }]$ (resp. $[\frac{ E_\n }{\tau_\n }, E_\n ]$).
\item $\omega_\n $ is drawn uniformly from $[1,10]$.
\item $\yy\nt$ is taken equal to $\ox\nt$ on the first time periods of $\T_\n $ (first available time periods) until reaching $ E_\n $ (consumption profile %
 ``plug and charge'').
\end{itemize}

We consider the following \emph{coupling constraints}:%
\begin{align}
\label{cst:ramp}
-50 \leq \xag_{T}-\xag_1 \leq 50 \\
\label{cst:capa} \xag_t \leq 1400 , \quad \forall \t \in \T
\end{align}
Constraint \eqref{cst:ramp} is a periodicity constraint which imposes that the aggregate demand $X_T$ at the end of time horizon $\T$ is  close to the aggregate demande  $X_1$ at the start of the horizon: the idea behind this constraint is to ensure that the production schedule computed for horizon $\T$ can be applied on a day-to-day, periodical basis.
Constraint \eqref{cst:capa} is a capacity constraint, 
 (e.g. induced by production capacities).%
 \newif\ifAlgoProj
\AlgoProjfalse
\ifAlgoProj
Let us write the coupling constraints  in the closed form:
$ \Amat \xxag \leq \bb , $
where $\Amat $ is a real matrix of size $ (T+2) \times T$  and $\bb\in \rit^{T+2}$.
\else
\fi
\paragraph{Computing populations with $k$-means}

To apply the clustering procedure of \Cref{subsec:class}, we use the $k$-means algorithm \cite{lloyd1982least}, where ``$k$''$=\popcard$ is the number of populations (groups) to replace the large set of $N$ players.
 For each player $ \n $, we define her parametric description vector:
\begin{equation}
\paramvec_\n  \eqd [ \omega_\n , \yy_\n , E_\n , \uux_\n , \oox_\n  ] \in \rit^{3T+2}\ .
\end{equation}
The $k$-means algorithm finds a partition $(S_\pidx)_{1\leq \pidx \leq \popcard}$ of $\N $ into $\popcard$ clusters, approximately solving the combinatorial problem:
\begin{equation*}
\hh \min_{S_1,\dots, S_\popcard} \hspace{-2pt} \sum_{1\leq \pidx \leq \popcard} \sum_{\paramvec \in S_\pidx} \hh \norm{ \mathbb{E}_{S_\pidx}\hspace{-2pt}(\paramvec) - \paramvec   }^2 \hh = \hspace{-5pt}\min_{S_1,\dots, S_\popcard}  \hh \hspace{-1pt} \sum_{1\leq \pidx \leq \popcard} \hspace{-6pt} |S_\pidx|  \mathbb{V}\text{ar}(S_\pidx) ,
\end{equation*}
where $\mathbb{E}_{S_\pidx}(\paramvec)=\frac{1}{|S_\pidx|}\sum_{\n\in S_\pidx} \paramvec_\n $ denotes the average value of $\paramvec$ over the set $S_\pidx$, taken to define $w_\pidx$, $\yy_\pidx$, $E_\pidx$, $\uux_\pidx$ and $\oox_\pidx$.

The simulations are run with different population numbers, with $\popcard  \ll N$  chosen among $\{5,10,20,50,100\}$.

\medskip 
The $k$-means algorithm enables to efficiently find  a heuristic and approximate solution of \eqref{eq:clusteringMinimizeK}. This algorithm does not search for the optimal choices of $(u_i)_i$ and $(f_i)_i$, as in the statement of problem  \eqref{eq:clusteringMinimizeK}.
Indeed, as the algorithm minimizes the squared distance of the average vector of parameters in $S_\pidx$ to the vectors of parameters of the points in $S_\pidx$, the solution obtained can be sub-optimal in terms of  $K(\mdset, \mduti)$.
Finding efficient methods to address  problem  \eqref{eq:clusteringMinimizeK} is an avenue for further research.

\paragraph{Computation methods}

We compute a VNE (\Cref{def:ve-finite}) with the original set of $N$ players and the approximating SVWE (\Cref{def:pseudoVNE}) as solutions of the associated GVI \eqref{cond:ind_opt_ve}.

We employ the iterative projection method  exposed in \cite[Algo. 2]{gentile2017nash}, adapted to our subdifferentiable case
by replacing the fixed step $\tau$ used in with a variable step $\tau^{(k)}=1/k$, in the spirit of subgradient algorithms  \cite{cohen1988auxiliary,auslender2009projected}.
\ifAlgoProj
 The principle of this algorithm is to relax the coupling constraint %
  and to consider the Lagrangian multipliers $\bm{\lambda} \in \rit_+^{T+2}$ associated to these constraints as extra variables, and to consider the \emph{extended} operator $T: \FX\epopset \times \rit_+^{T+2}  \rightrightarrows \rit^{IT} \times \rit^{T+2}$ defined as:
 \begin{equation*}
 T(\xx, \llam) = \begin{pmatrix}
 \big( \partial_1 f_\pidx(\xx_\pidx,\xxag) + \llam^\tr \Amat\big)_{i\in\I} \\
-( \Amat \xxag - \bb) \ 
 \end{pmatrix}
\end{equation*}  
on which we can apply a projected subgradient algorithm.  
 The advantage is that we can perform the projections on the sets $(\X_\n)_{n\in\N} $ and on $\rit_+^{T+2}$ as shown below:
\begin{algorithm}[H]
\begin{algorithmic}[1] %
\Require ${\xx}^{(0)}, \bm{\lambda}^{(0)}$, stopping criterion %
\State $k \leftarrow 0$  \;%
\While%
{stopping criterion not true}
\For{$\pidx =1$ to $\popcard$} \label{line:algo-for}
\State take $g_\pidx^{(k)} \in \partial_1 f_\pidx(\xx_\pidx^{(k)},\xxag^{(k)}) $ \; 
\State $\xx_\pidx^{(k+1)} \leftarrow \Proj_{\X_\pidx}\left(\xx_\pidx^{(k)}-\tau^{(k)}(g_\pidx^{(k)} + {\bm{\lambda}^{(k)}}^\tr \Amat ) \right) $ %
\label{line:algo-proj}\;
\EndFor
\State $\bm{\lambda}^{(k+1)} \hh\leftarrow \left( {\bm{\lambda}^{(k)} \hh-\hh \tau^{(k)} (b-2 \Amat\xxag^{(k+1)} + \Amat \xxag^{(k)} )} \right)^+$ \;
\State $  k \leftarrow k+1 $ \;
\EndWhile
\end{algorithmic}
\caption{Projected Descent Algorithm }
\label{algo:PGDFS}
\end{algorithm}

The stopping criterion that we adopt here is the distance between two iterates: the algorithm stops when $\|  (\lambda^{(k+1)},\xx^{(k+1)})-(\lambda^{(k)},\xx^{(k)}) \|_2 \leq 10^{-3}$.
Although the algorithm converges for this criterion in practice in our numerical experiments, the general convergence of \Cref{algo:PGDFS} is not proven theoretically.
Proving the convergence of a projected subgradient algorithm  for a general game, under \Cref{assp_convex_costs} and  assuming $\xx \mapsto  \big(\partial_1 f_\pidx(\xx_\pidx,\xxag)\big)_{ i\in\I }$ to be (strongly) monotone, is out of the scope of this paper, but would constitute an interesting path for further research.

Due to the form of the strategy sets considered \eqref{eq:set_Edemand}, the projection steps (Line 5) can be computed efficiently and exactly in $\mathcal{O}(T)$ with the Brucker algorithm \cite{brucker1984n}. However, if we consider more general strategy sets (arbitrary convex sets), this projection step can be costly: in that case, other algorithms such as \cite{fukushima1986relaxed} would be more efficient.                    
\else 
The algorithm stops when the Euclidean distance between two iterates is smaller than $10^{-3}$.
We refer to \cite{gentile2017nash} for more details on the algorithm.
\fi
\paragraph{A trade-off between precision and computation time}

Simulations were run using Python on a single core Intel Xeon @3.4Ghz and 16GB of RAM. 
In this simple example, the computation of the actual VNE $\hxx$ (and aggregate profile $\hxxag$) of the original game with $N=2000$ players is still possible, and is used to measure the precision of our approximations.
\begin{figure}[htb]
\hspace{-8pt}
     \centering
     \subfloat[][]{\includegraphics[width=0.48\columnwidth]{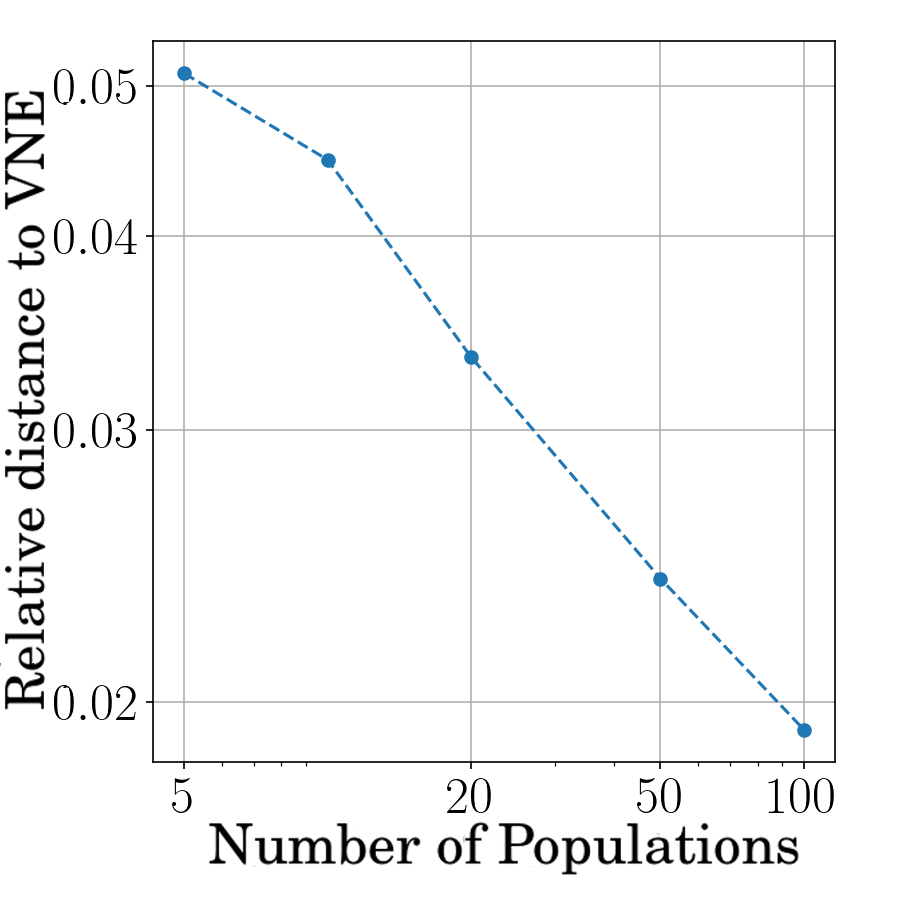%
         }\label{fig:cvg_pne_dists} }
     \subfloat[][]{\includegraphics[width=0.485\columnwidth]{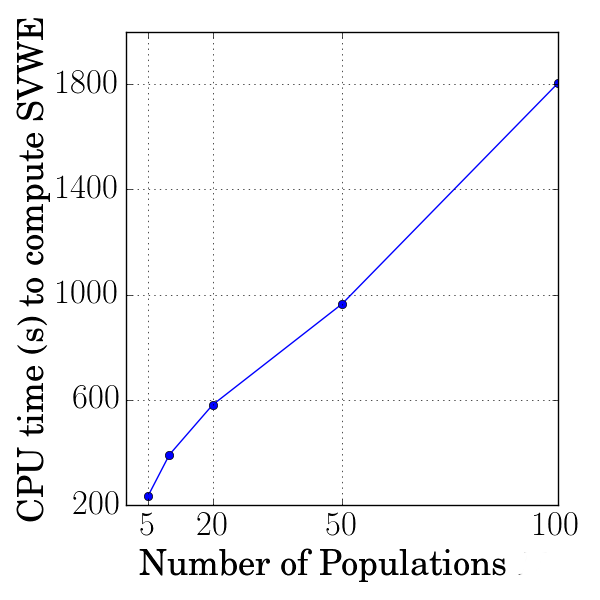}
     \label{fig:cvg_pne_CPU}}
     \caption{(a) Relative error to VNE ; (b) SVWE CPU Time.%
     }
\end{figure}

\Cref{fig:cvg_pne_dists} and \Cref{fig:cvg_pne_CPU} show  the two main metrics  to consider to choose a relevant number of populations $\popcard$: the precision of the SVWE approximating the equilibrium 
(measured by the relative distance of the  aggregate SVWE profile $\sxxag$ to the aggregate VNE  profile $\hxxag$  computed along, $\| \sxxag-\hxxag\|  /  \| \hxxag \| $), and the CPU time needed to compute the SVWE. 
We can observe from \Cref{fig:cvg_pne_dists} that the distance between the aggregate equilibrium profile and its estimation decreases with $\popcard$ at a rate of $\mathcal{O}(1/\popcard^a)$ where $a$ is estimated numerically as $a \approx 0.37$.
This is explained in light of \Cref{th:estimNE:main} which gives a bound depending on $\sqrt{K(\mdset,\mduti)}$ and in addition with the following remarks, which give the intuition that both $\mdset$ and $\mduti$ should evolve as $\mathcal{O}(1/\popcard)$ if parameters among players in $\N$ are distributed uniformly:
\begin{itemize}[wide]
\item the Hausdorff distance between  parameterized polyhedra is Lipschitz continuous w.r.t their parameter vectors (see  \cite{batson1987combinatorial});
\item similarly, as subgradients of utility functions are reduced to a point, one has, for all $n$,
\begin{align*}
\duti_\pidx = \max_{\n\in \N_\pidx} \max_{\xx \in \M} 2 \norm{ \omega_\pidx (\xx  - \yy_\pidx)- \omega_\n  (\xx-\yy_\n )} \\
 = \mathcal{O}\left(\max_{\n  \in\N_\pidx}  | \omega_\pidx-\omega_\n | + \norm{ \yy_\pidx- \yy_\n } \right) \ .
\end{align*}
\end{itemize}

\Cref{fig:cvg_pne_CPU} shows the CPU time needed to compute 
the SVWE with our stopping criterion.
The CPU time evolves linearly with the number of populations $\popcard$, which is  explained by the \emph{distributed} structure of the algorithm, requiring a number of iterations proportional to $\popcard$. 
The computation of  a solution of the clustering problem with the $k$-means algorithm is negligible (less than ten seconds  for each value of $\popcard$).

Computing the VNE of the original game with the same configuration and stopping criterion took  3 hours and 26 minutes, which is more than six times the time needed to compute the SVWE with one hundred populations.

Last, the error between the aggregate demand profile at equilibrium  and its approximation (\Cref{fig:cvg_pne_dists}) is between 2\% and 5\%, which remains significant. 
However, as pointed out in \Cref{sec:approxRes}, the quality of the approximation depends on the heterogeneity of the set of players $\N $:
as the parameters are drawn uniformly in this experiment, the set of players $\N $ presents a large variance, and provides us with a ``worst'' case (as opposed to  \Cref{ex:homogenousPops} which offers an optimal situation).

\section{Conclusion}

This paper shows that equilibria in large splittable aggregative games  can be approximated with a Wardrop Equilibrium of an auxiliary population game of smaller dimension. 
Our results give explicit bounds on the distance of this approximating equilibrium to the equilibria of the original large game. 
These theoretical results can be used in practice to solve, by an iterative method, complex nonconvex bilevel programs where the lower level is the equilibrium of a large aggregative game, for instance, to optimize tariffs or tolls for the operator of a network.
The analysis of such a procedure will constitute an interesting extension of this work.
Besides, this study highlighted the need for efficient algorithms for solving  clustering problems of type \eqref{eq:clusteringMinimizeK}, as well as distributed algorithms for solving GVI problems without a strong monotonicity assumption \cite{gentile2017nash,cohen1988auxiliary}. Those algorithmic considerations are also interesting avenues for further research.

\appendices

\crefalias{section}{appsec}

\section{Proof of \Cref{lm:intprofileAt}: Existence of interior profile}
\label{app:proof:lm-intprofile}

Let $\check{\xx}\in \FX$ be such that for all $ \n $:
\vspace{-0.15cm}
\begin{equation*}
d(\check{\xx}_\n , \rbd \X_\n )=\txt \max_{\xx\in \X_\n }d(\xx,\rbd \X_\n ) \eqd \eta_\n .
\vspace{-0.15cm}
\end{equation*}  
Denote $\bar{\check \xxag}\hh =\hh \tfrac{1}{N}\sum_\n  \check{\xx}_\n $ and 
$\eta \hh=\hh \min_{\n} \eta_\n >0$. 
Let $\yy\in\FX(\Aset)$ and $\byyag=\tfrac{1}{N}\sum_\n  \yy_\n $ be s.t.:
\begin{equation*} 
d(\byyag, \rbd \Aset)=\txt\max_{\bxxag\in \Sxag\cap \Aset}d(\bxxag,  \rbd \Aset) \ . \end{equation*}
Let us denote $t= {d(\byyag, \rbd \Aset)}/{3 \diamX}$ and  let us define $\zz =\yy - t(\yy -\check{ \xx})\in \FX$  and  $\bzzag=\tfrac{1}{N}\sum_\n  \zz_\n $.
We obtain:
\begin{equation*}
\|\byyag- \bzzag\|=t\|\byyag-\bar{\check \xxag}\|\leq t 2 \diamX\leq \tfrac{2}{3}d(\byyag, \rbd \Aset) \ ,
\end{equation*}  
hence $\bzzag \in \Sxag\cap \rlt  \Aset$, where $\rlt$ means the relative interior. 
Besides, for any $\n $, $\zz_\n  =\yy_\n  - t(\yy_\n  - \check{\xx}_\n ) $. Since $d(\check{\xx}_\n , \rbd \X_\n )\geq \eta$, $\yy_\n  \in \X_\n $, and $\X_\n $ is convex, we have:
\begin{equation*}
d(\zz_\n , \rbd \X_\n  ) \geq \eta t = \tfrac{\eta}{3 \diamX}d(\byyag, \rbd \Aset) \ .
\end{equation*}
We can conclude by defining $\rho \eqd \frac{\eta}{3 \diamX}d(\byyag, \rbd \Aset)$. 

\section{Proof of the lemmas of  \Cref{th:estimNE:main}}
\label{app:proof:main}

\begin{proof} \textit{\Cref{lm:FYAt}}
i) Suppose  $\xx \notin \X_\n $. Let $\yy \eqd \Proj_{\X_\n }(\xx)\neq \xx$. As $\yy \in \aff \X_\n  \subset \aff \X_\pidx $, then $\xx-\yy \in \aff \X_\pidx$. Let $\zz \eqd \xx +  \dset_\pidx \frac{ \xx-\yy}{\norm{\xx-\yy}}$. Then, $\zz \in \X_\pidx$  because $\norm{\zz - \xx } \leq \dset_\pidx$.
By the convexity of $\X_\n $ and the definition of $\yy$, we have $d(\zz,\X_\n )=  d(\xx,\X_\n )+ \dset_\pidx> \dset_\pidx$,  contradicting the fact that $\dset_\pidx \geq  d_H(\X_\n ,\X_\pidx)$.

The proof of ii) is symmetric and omitted.
\end{proof}

\begin{proof} \textit{\Cref{lem:dist_agg_genized_sets_at}}~\\
i) For $\xx \in \FX\epopset(\Aset)$, define $\ww\in \FX$ as follows: $\forall \pidx\in\popset$, $\forall \n \in \N_\pidx$, let $\ww_\n  \eqd \xx_\pidx + t(\zz_\n  - \xx_\pidx )$ where $\zz$ is defined in \Cref{lm:intprofileAt}, with $t\eqd 2\mdset/\rho<1$. %
On the one hand:
\begin{align*}
&\forall \pidx\in\popset,\ \forall \n \in \N_\pidx, \ d(\zz_\n , \rbd \X_\pidx) \geq \rho - \mdset  \quad \\
\text{ implies } \quad  & d(\ww_\n ,\rbd \X_\pidx) \geq t(\rho-\mdset)> t\rho/2 =\mdset \ ,
\end{align*} 
because each point in the ball with radius $t( \rho-\mdset)$ centered at $\ww_\n $ is on the segment linking $\xx_\pidx$ and some point in the ball with radius $\rho-\mdset$ centered at $\zz_\n $ which is contained in $\X_\n $.

 Thus, $\ww_\n  \in \X_\n $  $\forall \n \in \N_\pidx$ according to \Cref{lm:FYAt}.i). 
On the other hand, the linear mapping $S: \rit^{IT} \ni \bv \mapsto \tfrac{1}{N}\sum_{\n\in \N } \bv_\n $ maps the segment linking $\psi(\xx)$ and $\zz$ in $\FX(\Aset)$ to a segment linking $\bxxag=\tfrac{1}{N}\sum_\pidx \popcard_\pidx \xx_\pidx$ and $\bzzag$ in the convex $ \Aset$. We get:
 \begin{equation*}
 \txt\tfrac{1}{N}\sum_{\n\in \N } \ww_\n  = t\bzzag + (1-t)\bxxag \in \Aset 
 \end{equation*} 
and thus $\ww\in \FX(\Aset)$. 
Finally, $\|\ww_\n  - \psi_\n (\xx)\| =t\|\zz_\n -\psi_\n (\xx)\| \leq t2 \diamX=  4\diamX \tfrac{\mdset}{\rho}$.

ii) For $\xx\in \FX(\Aset)$, let $\yy\eqd \xx+t(\zz-\xx)$ with $t\eqd \frac{\mdset}{\rho}$. 
Then, by similar arguments as above, $d(\yy_\n , \rbd \X_\n ) \geq \mdset$ hence $\yy_\n \in \X_\pidx$ and $\bpsi(\yy)\in \FX\epopset$. 
From the convexity of $\Aset$:
\begin{equation*}
\txt\tfrac{1}{N}\sum_\n  \yy_\n =t\bzzag+(1-t)(\tfrac{1}{N}\sum_\n  \xx_\n ) \in \Aset \ . 
\end{equation*}
  Hence $\ww\eqd \bpsi(\yy)\in \FX\epopset(\Aset)$. Finally, $\|\ww_\pidx-\psi_\pidx(\xx)\|=t\|\sum_{\n\in \popcard_\pidx}(\zz_\n -\xx_\n )\|\leq 2 \diamX \popcard_\pidx\frac{\mdset}{\rho}$.
\end{proof}

\begin{small}
\bibliographystyle{IEEEtran}
\bibliography{../../../bib/shortJournalNames,../../../bib/biblio1,../../../bib/biblio2,../../../bib/biblio3,../../../bib/biblio4,../../../bib/biblioBooks}
\end{small}
\end{document}